\documentclass[11pt]{article}
\usepackage{latexsym}
\usepackage{amsfonts,amssymb}
\usepackage{amsmath,amssymb,amsthm,xspace,fullpage}
\usepackage{graphicx}
\usepackage{enumerate}
\usepackage{color}

\newtheorem{theorem}{Theorem}
\newtheorem{lemma}{Lemma}

\newtheorem{proposition}{Proposition}

\newtheorem{property}{Property}

\usepackage{float}
\usepackage[absolute,overlay]{textpos}
\usepackage[ruled, vlined, linesnumbered]{algorithm2e}
\usepackage{here}

\newcommand{\argmin}{\mathop{\rm arg\,min}\limits}
\newcommand{\hide}[1]{}

\def\squarebox#1{\hbox to #1{\hfill\vbox to #1{\vfill}}}

 \newcommand{\bs}{\bigskip} 
  
 \newcommand{\hs}[1]{\hspace*{ #1 mm}} 

\newcommand{\ignore}[1]{}

\allowdisplaybreaks[1]

\begin{document}
%%%%%%%%%%%%%%%%%%
\pagestyle{plain}
\begin{center}
{\Large {\bf 
Minmax Regret 1-Sink Location Problems on Dynamic Flow Path Networks with Parametric Weights
}}
%\footnote{} 
\bs\\

{\sc Tetsuya Fujie}$^1$ \hspace{5mm} 
{\sc Yuya Higashikawa}$^1$ \hspace{5mm} 
{\sc Naoki Katoh}$^1$ \hspace{5mm} 
{\sc Junichi Teruyama}$^1$ \hspace{5mm} 
{\sc Yuki Tokuni}$^2$ \hspace{5mm} 

\

$^1${School of Social Information Science, University of Hyogo, Japan}; \\
{\tt \{fujie, higashikawa,naoki.katoh,junichi.teruyama\}@sis.u-hyogo.ac.jp}\\

$^2${School of Science and Technology, Kwansei Gakuin University, Japan}; \\
{\tt enj32048@kwansei.ac.jp}

\end{center}
\bs

\begin{abstract}
This paper addresses the minmax regret 1-sink location problem on dynamic flow path networks with parametric weights.
We are given a {\em dynamic flow network} consisting of
an undirected path with positive edge lengths, 
positive edge capacities, and nonnegative vertex weights. 
A path can be considered as a road, 
an edge length as the distance along the road and 
a vertex weight as the number of people at the site. 
An edge capacity limits the number of people that can enter the edge per unit time. 
We consider the problem of locating a \textit{sink} in the network, to which all the people evacuate from the vertices as quickly as possible.
In our model, 
each weight is represented by a linear function in a common parameter $t$,
and the decision maker who determines the location of a sink does not know the value of $t$. 
We formulate the sink location problem under such uncertainty as the \textit{minmax regret problem}. 
Given $t$ and a sink location $x$,
the cost of $x$ under $t$ is the sum of arrival times at $x$ 
for all the people determined by $t$.
The regret for $x$ under $t$ is the gap between the cost of $x$ under $t$ and the optimal cost under $t$. 
The task of the problem is formulated as the one to find a sink location that minimizes the maximum regret over all $t$. 
For the problem, we propose
an $O(n^4 2^{\alpha(n)} \alpha(n) \log n)$ time algorithm
where $n$ is the number of vertices in the network and 
$\alpha(\cdot)$ is the inverse Ackermann function. 
Also for the special case in which every edge has the same capacity,
we show that the complexity can be reduced to $O(n^3 2^{\alpha(n)} \alpha(n) \log n)$.
\end{abstract}

%%%%%%%%%%%%%%%%%%%%%%%%%%%%%%%%%%%%%%%%
% Section 1
%%%%%%%%%%%%%%%%%%%%%%%%%%%%%%%%%%%%%%%%

\section{Introduction}\label{sec:intro}

Recently, many disasters, such as earthquakes, nuclear plant accidents, volcanic eruptions and flooding,
have struck in many parts of the world,
and it has been recognized that orderly evacuation planning is urgently needed.
A powerful tool for evacuation planning is the {\em dynamic flow model} introduced by Ford and Fulkerson~\cite{ford1958},
which represents movement of commodities over time in a network.
In this model,
we are given a graph 
with \textit{source} vertices and \textit{sink} vertices.
Each source vertex is associated with a positive weight, called a \textit{supply},
each sink vertex is associated with a positive weight, called a \textit{demand}, and each edge is associated with positive length and capacity.
An edge capacity limits the amount of supply that can enter the edge per unit time.
One variant of the dynamic flow problem is the {\em quickest transshipment problem}, of which
the objective is to send exactly the right amount of supply out of sources into sinks with satisfying the demand constraints
in the minimum overall time.
Hoppe and Tardos~\cite{hoppe2000} provided a polynomial time algorithm for this problem in the
case where the transit times are integral.
However, the complexity of their algorithm is very high.
Finding a practical polynomial time solution to this problem is still open.
A reader is referred to a recent survey by Skutella~\cite{skutella2009} on dynamic flows.

This paper discusses a related problem,
called the {\em sink location problem}~\cite{belmonte2015polynomial,BenkocziBHKK18,BenkocziBHKK19,BhattacharyaGHK17,chen2016sink,chen2018minmax,higashikawa2014f,HigashikawaGK15,mamada2006},
of which 
the objective is to find a location of sinks in a given dynamic flow network so that all the supply is sent to the sinks as quickly as possible.
For the optimality of location, the following two criteria can be naturally considered: the minimization of {\em evacuation completion time} and {\em aggregate evacuation time}
(i.e., {\em sum of evacuation times}).
We call the sink location problem that requires finding a location of sinks on a dynamic flow network
that minimizes the evacuation completion time (resp. the aggregate evacuation time) 
the {\sf CTSL} problem (resp. the {\sf ATSL} problem).
Several papers have studied the {\sf CTSL} problems~\cite{belmonte2015polynomial,BhattacharyaGHK17,chen2016sink,chen2018minmax,higashikawa2014f,HigashikawaGK15,mamada2006}.
On the other hand, 
for the {\sf ATSL} problems, 
we have a few results only for path networks~\cite{BenkocziBHKK18,BenkocziBHKK19,HigashikawaGK15}.

In order to model the evacuation behavior of people,
it might be natural to treat each supply as a discrete quantity
as in~\cite{hoppe2000,mamada2006}.
Nevertheless, almost all the previous papers on sink location problems~\cite{belmonte2015polynomial,BhattacharyaGHK17,chen2016sink,chen2018minmax,higashikawa2014f,HigashikawaGK15} treat 
each supply as a continuous quantity
since 
it is easier for mathematically handling the problems and
the effect of such treatment is small enough to ignore
when the number of people is large.
Throughout the paper, we adopt the model with continuous supplies.

Although the above two criteria are reasonable, they may not be practical
since the population distribution is assumed to be fixed.
In a real situation, the number of people in an area may vary depending on the time,
e.g., in an office area in a big city, there are many people during the daytime on weekdays while there are much less people on weekends or during the night time. 
In order to take such the uncertainty into account, 
Kouvelis and Yu~\cite{kouvelis1997} introduced the {\em minmax regret model}.
In the {\em minmax regret sink location problems}, 
we are given a finite or infinite set $S$ of {\em scenarios},
where each scenario gives a particular assignment of weights on all the vertices.
Here, for a sink location $x$ and a scenario $s \in S$, 
we denote the evacuation completion time or aggregate evacuation time by $F(x,s)$. 
Then, the problem can be understood as a 2-person Stackelberg game as follows. 
The first player picks a sink location $x$ and the second player chooses a scenario $s \in S$
that maximizes the regret defined as $R(x,s) := F(x,s)-\min_x F(x,s)$.
The objective of the first player is to choose $x$ that minimizes the maximum regret.
Throughout the paper, we call the minmax regret sink location problem, where the regret is defined with the evacuation completion time (resp. the aggregate evacuation time), 
the {\sf MMR-CTSL} problem (resp. the {\sf MMR-ATSL} problem).
The {\sf MMR-CTSL} problems have been studied so far~\cite{arumugam2019minmax,bhattacharya2015improved,golin2018minmax,higashikawa2015a,higashikawa2014f,li2016a,li2016b}. 
On the other hand, 
for the {\sf MMR-ATSL} problems, 
we have few results~\cite{bhattacharya2018n,higashikawa2018minimax}
although the problems are also important theoretically and practically.

As for how to define a set of scenarios, 
all of the previous studies on the minmax regret sink location problems adopt the model with {\em interval weights},
in which each vertex is given the weight as a real interval,
and a scenario is defined by choosing an element of the Cartesian product of all the weight intervals over the vertices.
One drawback of the minmax regret model with interval weights is that 
each weight can take an independent value, thus we consider some extreme scenarios which may not happen in real situations, e.g, a scenario where all the vertices have maximum weights or minimum weights.
To incorporate the dependency among weights of all the vertices into account, we adopt the model with {\em parametric weights} (first introduced by Vairaktarakis and Kouvelis~\cite{vairaktarakis1999incorporation} for the minmax regret median problem),
in which each vertex is given the weight as a linear function in a common parameter $t$ on a real interval,
and a scenario is just determined by choosing $t$.
Note that considering a real situation,
each weight function should be more complex,
however, such a function can be approximated by a piecewise linear function.
Thus superimposing all such piecewise linear functions, it turns out that for a sufficiently small subinterval of $t$, every weight function can be regarded as linear, and by solving multiple subproblems with linear weight functions, we can obtain the solution.

In this paper, 
we study the {\sf MMR-ATSL} problem on dynamic flow path networks with parametric weights.
Our main theorem is below.

\begin{theorem}[Main Results]\label{thm:main}
Suppose that we are given a dynamic flow path network of $n$ vertices with parametric weights.
\begin{itemize}
\item[(i)] The {\sf MMR-ATSL} problem can be solved 
in time $O( n^4 2^{\alpha(n)} \alpha(n) \log n)$, 
where $\alpha(\cdot)$ is the inverse Ackermann function.

\item[(ii)] When all the edge capacities are uniform, 
the {\sf MMR-ATSL} problem can be solved 
in time $O( n^3 2^{\alpha(n)} \alpha(n) \log n)$.
\end{itemize}
\end{theorem}

Note that 
the {\sf MMR-ATSL} problem with interval weights
is studied by~\cite{bhattacharya2018n,higashikawa2018minimax},
and only for the case with the uniform edge capacity,
Higashikawa~et~al.~\cite{higashikawa2018minimax} provide an $O(n^3)$ time algorithm,
which is improved to one running in $O(n^2 \log^2 n)$ time by~\cite{bhattacharya2018n}.
However, for the case with general edge capacities,
no algorithm has been known so far.
Therefore, our result implies that 
the problem becomes 
solvable in polynomial time
by introducing parametric weights.

The rest of the paper is organized as follows.
In Section~\ref{sec:preliminary}, we give the notations and the fundamental properties that are used throughout the paper.
In Section~\ref{sec:keyidea}, we give the key lemmas and the algorithms that solves the problems, which concludes the paper.

%%%%
%Section:Preliminaries
%%%%
\section{Preliminaries}\label{sec:preliminary}
For two real values $a, b$ with $a<b$, let $[a, b] = \{t \in {\mathbb R} \mid a  \leq t \leq b\}$,
$(a, b) = \{t \in {\mathbb R} \mid a  < t < b\}$,
and 
$(a, b] = \{t \in {\mathbb R} \mid a  < t \leq b\}$.

In our problem, we are given a real interval $T = [t^{-}, t^{+}] \subset {\mathbb R}$ and 
a dynamic flow path network ${\cal P} = (P, {\bf w}(t), {\bf c}, {\bf l}, \tau)$, 
which consists of five elements: 
$P = (V, E)$ is a path with vertex set $V = \{v_i \mid 1 \leq i \leq n\}$ and 
edge set $E = \{e_i = (v_i, v_{i+1}) \mid 1 \leq i \leq n-1\}$, 
${\bf w}(t)$ is a vector $\langle w_1(t), \ldots, w_n(t) \rangle$ of which 
component $w_i(t)$ is a {\it weight function} $w_i : T \rightarrow {\mathbb R}_{\geq 0}$ 
which is linear in a parameter $t$ and nonnegative for any $t \in T$, 
a vector ${\bf c} = \langle c_1, \ldots, c_{n-1} \rangle$ consists of the {\it capacity} $c_i$ of edge $e_i$, 
a vector ${\bf l} = \langle \ell_1, \ldots, \ell_{n-1} \rangle$ consists of the {\it length} $\ell_i$ of edge $e_i$, 
and $\tau$ is the {\it time} which the supply takes to move a unit distance on any edge. 
Let us explain how edge capacities and lengths affect the evacuation time.
Consider an evacuation under fixed $t \in T$.
Suppose that at time $0$, the amount $w$ of supply is at vertex $v_{i+1}$ and going through edge $e_i$ towards vertex $v_i$.
The first fraction of supply from $v_{i+1}$ can arrive at $v_i$ at time $\tau \ell_i$.
The edge capacity $c_i$ represents the maximum amount of supply which can enter $e_i$ in a unit time interval, 
so all the supply $w$ can complete leaving $v_{i+1}$ at time $w/c_i$.
Therefore, all the supply $w$ can complete arriving at $v_{i}$ at time $\tau \ell_i+w/c_i$.

For any integers $i,j$ with $1 \leq i \le j \leq n$, we denote the sum of weights from $v_i$ to $v_j$ by $W_{i,j}(t) = \sum_{h=i}^{j} w_h(t)$. 
For the notation, we define $W_{i,j}(t) = 0$ for $i,j$ with $i > j$. 
For a vertex $v_i \in V$, we abuse $v_i$ to denote the distance between $v_1$ and $v_i$,
i.e., $v_i = \sum_{j=1}^{i-1} \ell_j$.
For an edge $e_i \in E$, we abuse $e_i$ to denote a real open interval $(v_i,v_{i+1})$.
We also abuse $P$ to denote a real closed interval $[0,v_n]$.
If a real value $x$ satisfies $x \in (v_i, v_{i+1})$, 
$x$ is said to be a point on edge $e_i$
to which the distance from $v_i$ is value $x-v_i$.
Let $C_{i,j}$ be the minimum capacity for all the edges from $e_i$ to $e_j$, i.e., $C_{i,j} = \min\{c_{h} \mid i \leq h \leq j \}$.

Note that we precompute values $v_i$ and $W_{1,i}(t)$ for all $i$ in $O(n)$ time,
and then, $W_{i,j}(t)$ for any $i,j$ can be obtained in $O(1)$ time as $W_{i,j}(t)=W_{1,j}(t)-W_{1,i-1}(t)$.
In addition, $C_{i,j}$ for any $i,j$ can be obtained in $O(1)$ time 
with $O(n)$ preprocessing time, which is known as the {\it range minimum query}~\cite{Alstrup2002,Bender2000}.

\subsection{Evacuation Completion Time on a Dynamic Flow Path Network}

In this section, we see the details of evacuation phenomenon using a simple example,
and eventually show the general formula of evacuation completion time on a path, 
first provided by Higashikawa~\cite{Higashikawa14}.
W.l.o.g., an evacuation to a sink $x$ follows the first-come first-served manner at each vertex,
i.e., when a small fraction of supply arrives at a vertex $v$ on its way to $x$, 
it has to wait for the departure if there already remains some supply waiting for leaving $v$.

Let us consider an example with $|V|=3$ where $V=\{v_1, v_2, v_3\}, E=\{e_1=(v_1, v_2), e_2=(v_2,v_3)\}$.
Assume that the sink $x$ is located at $v_1$, and
under a fixed parameter $t \in T$,
the amount of supply at $v_i$ is $w_i$ for $i=2,3$.

All the supply $w_1$ at $v_1$ immediately 
completes its evacuation at time $0$
and we send all the supply $w_2$ and $w_3$ to $v_1$ as quickly as possible. 
Let us focus on how the supply of $v_3$ moves to $v_1$.
First, the foremost fraction of supply from $v_3$ arrives at $v_2$ at time $\tau \ell_2$,
and all the supply $w_3$ completes leaving $v_3$ at time $w_3/c_2$,
i.e., it completes arriving at $v_2$ at time $\tau \ell_2 + w_3/c_2$.
Suppose that at time $\tau \ell_2 + w_3/c_2$, the amount $w' (\ge 0)$ of supply remains at $v_2$.
From then on, 
the time required to send all the supply $w'$ to $v_1$ is $\tau \ell_1 + w'/c_2$.
Thus, the evacuation completion time is expressed as
\begin{equation}\label{eq:example1}
\tau (\ell_1+\ell_2) + \frac{w_3}{ c_2 } + \frac{w'}{c_1}.
\end{equation}
We observe what value $w'$ takes in the following cases. 

{\bf Case 1: It holds \mbox{\boldmath $c_1 \ge c_2$}.}
In this case, the amount of supply at $v_2$ should be non-increasing, 
because the amount $c_1$ of supply leaves $v_2$
and the amount at most $c_2$ of supply arrives at $v_2$ per unit time.
Let us consider the following two situations at time $\tau \ell_2 + w_3/c_2$: 
When all the supply $w_3$ completes arriving at $v_2$, 
there remains no supply at $v_2$, that is, $w' = 0$ holds or not. 
If $w'=0$ holds, then substituting it into~\eqref{eq:example1}, 
the evacuation completion time is expressed as
\begin{equation}\label{eq:example2}
\tau (\ell_1+\ell_2) + \frac{w_3}{ c_2 }.
\end{equation}
Otherwise, that is $w' > 0$ holds, 
there remains a certain amount of supply at $v_2$ even at time $\tau \ell_2$ 
since the amount of supply at $v_2$ is non-increasing. 
Thus at time $\tau \ell_2$, 
the amount $w_2-c_1 \tau \ell_2$ of supply remains at $v_2$.
From time $\tau \ell_2$ to time $\tau \ell_2 + w_3/c_2$, 
the amount of supply waiting at $v_2$ decreases by $c_1-c_2$ per unit time.
Then, we have
\[
w' 
= w_2 - c_1(\tau \ell_2) - (c_1-c_2) \cdot \frac{w_3}{c_2}
= w_2 + w_3 - c_1 \tau \ell_2 - \frac{c_1 w_3}{c_2}.
\]
Thus, the evacuation completion time is expressed as
\begin{equation}\label{eq:example3}
%\tau \ell_2 + \frac{w_3}{c_2} + \tau \ell_1 + \frac{w_2-c_1(\tau \ell_2)-(c_1-c_2)(w_3/c_2)}{c_1} = 
\tau (\ell_1 + \ell_2) + \frac{w_3}{c_2} + \frac{w_2 + w_3 - c_1\tau \ell_2 - c_1 w_3 / c_2}{c_1} = 
\tau \ell_1 + \frac{w_2 + w_3}{ c_1 }.
\end{equation}

{\bf Case 2: It holds \mbox{\boldmath $c_1 < c_2$}.}
In this case, the amount of supply waiting at $v_2$ increases by $c_2-c_1$ per unit time
from time $\tau \ell_2$ 
(when the foremost supply from $v_3$ arrives at $v_2$)
to time $\tau \ell_2 + w_3/c_2$
(when the supply from $v_3$ completes to arrive at $v_2$).
Let us consider the following two situations at time $\tau \ell_2$.
When the foremost supply from $v_3$ arrives at $v_2$, 
there remains no supply at $v_2$ or not. 

If there remains no supply at $v_2$ at time $\tau \ell_2$, 
then it holds $w'=(c_2-c_1)(w_3/c_2) = w_3 - c_1 w_3/c_2$ in~\eqref{eq:example1}. 
Thus, the evacuation completion time is expressed as
\begin{equation}\label{eq:example4}
\tau (\ell_1 + \ell_2) + \frac{w_3}{c_2} + %\frac{(c_2-c_1)(w_3/c_2)}{c_1} = 
\frac{w_3 - c_1 w_3/c_2}{c_1} = 
\tau (\ell_1+\ell_2) + \frac{w_3}{ c_1 }.
\end{equation}
Otherwise, the situation is similar to the latter case of Case~1. 
The difference is that the amount of supply waiting at $v_2$ increases by $c_2-c_1$ per unit time during from time $\tau \ell_2$ to time $\tau \ell_2 + w_3/c_2$,
while in Case 1, it decreases by $c_1-c_2$ per unit time. 
For this case, the evacuation completion time is given by formula~\eqref{eq:example3}. 

In summary of formulae~\eqref{eq:example2}--\eqref{eq:example4}, the evacuation completion time 
for a dynamic flow path network with three vertices 
is given by the following formula:
\begin{equation}\label{eq:formula_3}
%\max\left\{\frac{w_2+w_3}{c_1}+\tau\cdot \ell_1, \frac{w_3}{\min\{c_1, c_2\}}+\tau\cdot (\ell_1+\ell_2)\right\}.
\max\left\{
\tau \ell_1 + \frac{w_2+w_3}{c_1}, 
\tau (\ell_1+\ell_2) + \frac{w_3}{\min\{c_1, c_2\}}\right\}.
\end{equation}

Let us turn to the case with $n$ vertices, 
that is, $V = \{v_i \mid 1 \leq i \leq n\}$.
When the sink is located at $v_1$ and a parameter $t \in T$ is fixed, 
generalizing formula~\eqref{eq:formula_3}, 
the evacuation completion time 
is given by the following formula, which is provided by Higashikawa~\cite{Higashikawa14}:
\begin{eqnarray}\label{eq:formula_n}
\max_{2 \leq i \leq n}
\left\{
%\tau (\ell_1 + \cdots + \ell_{i-1}) + 
\tau \sum_{j=1}^{i-1} \ell_{j} + 
%\frac{w_i(t) + \cdots + w_n(t)}{ \min\{c_1, \ldots, c_{i-1}\}}
\frac{ \sum_{j=i}^{n}w_j(t)}{ \min_{1 \le j \le i-1} c_j}
\right\}
= 
\max_{2 \leq i \leq n}
\left\{
\tau v_{i} + \frac{W_{i,n}(t)}{C_{1,i}}
\right\}.
\end{eqnarray}
%Recall that $C_{1,i} = \min\{ c_h \mid 1 \leq h \leq i\}$. 
An interesting observation is that 
each $\tau v_{i} + W_{i,n}(t)/C_{1,i}$ in~\eqref{eq:formula_n}
is equivalent to the evacuation completion time 
for the transformed input so that only $v_i$ is given supply $W_{i,n}(t)$
and all the others are given zero supply. 

Let us give explicit formula of the evacuation completion time 
for fixed $x \in P$ and parameter $t \in T$.
Suppose that a sink $x$ is on edge $e_i = (v_{i}, v_{i+1})$. 
In this case, all the supply on the right side (i.e., at $v_{i+1}, \ldots, v_n$) will flow left to sink $x$ and 
all the supply on the left side (i.e., at $v_{1}, \ldots, v_{i}$) will flow right to sink $x$. 
First, we consider the evacuation for the supply on the right side of $x$. 
Supply on the path is viewed as a continuous value, and we regard that all the supply on 
the right side of $x$ is mapped to the interval $(0, W_{i+1,n}(t)]$. 
The value $z$ satisfying $z \in (W_{i+1,j-1}(t), W_{i+1,j}(t)]$ with $i+1 \le j\le n$ 
represents all the supply at vertices 
$v_{i+1}, v_{i+2}, \ldots, v_{j-1}$ plus partial supply of 
$z - W_{i+1,j-1}(t)$ at $v_j$. 
Let $\theta^{e_i}_{R} (x, t, z)$ denote
the time at which the first $z$ amount of supply on the right side of $x$
(i.e., $v_{i+1}, v_{i+2}, \ldots, v_n$) completes its evacuation to sink $x$. 
Modifying formula~\eqref{eq:formula_n}, 
$\theta^{e_i}_{R} (x, t, z)$ is given by the following formula:
For $z \in (W_{i+1,j-1}(t), W_{i+1,j}(t)]$ with $i+1 \le j\le n$,
\begin{eqnarray}\label{eq:ct_right_1}
\theta^{e_i}_{{\rm R}} (x, t, z) & = & 
\max_{i+1 \le h \le j}
\left\{
\tau (v_h - x) + \frac{z - W_{i+1,h-1}(t)}{C_{i, h}} \right\}.
\end{eqnarray}
In a symmetric manner, we consider the evacuation for the supply on the left side of $x$ (i.e., $v_1, \ldots, v_i$). 
The value $z$ satisfying $z \in (W_{j+1,i}(t),W_{j,i}(t)]$ with $1 \le j\le i$ represents 
all the supply at vertices $v_i, v_{i-1}, \ldots, v_{j+1}$ plus partial supply of $z - W_{j+1,i}(t)$ at $v_j$. 
Let $\theta^{e_i}_{L} (x, t, z)$ denote 
the time at which the first $z$ amount of supply on the left side of $x$ 
completes its evacuation to sink $x$, 
which is given by the following formula: 
For $z \in (W_{j+1,i}(t),W_{j,i}(t)]$ with $1 \le j\le i$,
\begin{eqnarray}\label{eq:ct_left}
\theta^{e_i}_{{\rm L}} (x, t, z) & = & 
\max_{j \le h \le i}
\left\{  \tau (x - v_h) 
+ \frac{z - W_{h+1,i}(t)}{C_{h,i}}
\right\}.
\end{eqnarray}
Let us turn to the case that sink $x$ is at a vertex $v_i \in V$.
We confirm that the evacuation times when the amount $z$ of supply originating 
from the right side of and the left side of $v_i$ to sink $v_i$
are given by $\theta^{e_{i}}_{{\rm R}} (v_i, t, z)$ and 
$\theta^{e_{i-1}}_{{\rm L}}(v_i, t, z)$, respectively.

\subsection{Aggregate Evacuation Time}
Let $\Phi(x, t)$ be the aggregate evacuation time (i.e., sum of evacuation time) 
when a sink is at a point $x \in P$ and the weight functions are fixed by a parameter $t \in T$. 
For a point $x$ on edge $e_i$ and a parameter $t \in T$, 
the aggregate evacuation time $\Phi(x,t)$ is defined by 
the integrals of the evacuation completion times 
$\theta^{e_i}_{{\rm L}} (x, t, z)$ over $z \in [0,W_{1,i}(t)]$
and $\theta^{e_i}_{{\rm R}} (x, t, z)$ over $z \in [0,W_{i+1,n}(t)]$,
i.e., 
\begin{eqnarray}\label{eq:at}
\Phi(x, t) = \int_{0}^{W_{1,i}(t)} \theta^{e_i}_{{\rm L}} (x, t, z) dz + \int_{0}^{W_{i+1,n}(t)} \theta^{e_i}_{{\rm R}} (x, t, z) dz.
\end{eqnarray}
In a similar way, if a sink $x$ is at vertex $v_i$, then 
$\Phi(v_i, t)$ is given by 
\begin{eqnarray}\label{eq:at_v}
\Phi(v_i, t) = \int_{0}^{W_{1, i-1}(t)} \theta^{e_{i-1}}_{{\rm L}} (v_i, t, z) dz + \int_{0}^{W_{i+1,n}(t)} \theta^{e_i}_{{\rm R}} (v_i, t, z) dz.
\end{eqnarray}

\subsection{Minmax Regret Formulation}\label{subsec:minmax_regret}

We denote by ${\rm Opt}(t)$ the minimum aggregate evacuation time with respect to a parameter $t \in T$. 
Higashikawa et al.~\cite{HigashikawaGK15} and Benkoczi et al.~\cite{BenkocziBHKK19}
showed that for the minsum $k$-sink location problems, 
there exists an optimal $k$-sink such that all the $k$ sinks are at vertices. 
This implies that we have
\begin{eqnarray}\label{eq:opt} 
%{\rm Opt}(t) = \min \{ \Phi(x, t) \mid x \in V \} 
{\rm Opt}(t) = \min_{x \in V} \Phi(x, t)
\end{eqnarray} 
for any $t \in T$. 
For a point $x \in P$ and a value $t \in T$, 
a {\it regret} $R(x,t)$ with regard to $x$ and $t$ is 
a gap between $\Phi(x,t)$ and ${\rm Opt}(t)$
that is defined as
\begin{eqnarray}\label{def:regret}
R(x, t) = \Phi(x,t) - {\rm Opt}(t).
\end{eqnarray}
The {\it maximum regret} for a sink $x \in P$, 
denoted by $MR(x)$, 
is the maximum value of $R(x,t)$ with respect to $t \in T$. Thus, $MR(x)$ is defined as
\begin{eqnarray}\label{def:maxregret}
% MR(x) = \max\{ R(x,t) \mid t \in T \}.
MR(x) = \max_{t \in T} R(x,t).
\end{eqnarray}
Given a dynamic flow path network  ${\cal P}$ and a real interval $T$, the problem 1-MMR-AT-PW is defined as follows:
\begin{equation}\label{problem}
\mbox{ minimize } MR(x) \mbox{ subject to  } x\in P
\end{equation}
Let $x^*$ denote an optimal solution of~\eqref{problem}.

\subsection{Piecewise Functions and Upper/Lower Envelopes}\label{subsec:envelopes}

A function $f: X (\subset {\mathbb R}) \rightarrow {\mathbb R}$ is called a {\em piecewise polynomial function}
if and only if real interval $X$ can be partitioned into subintervals $X_1, X_2, \ldots, X_m$ 
so that $f$ forms as a polynomial $f_i$ on each $X_i$. 
We denote such a piecewise polynomial function $f$ 
by $f = \langle (f_1, X_1), \ldots, (f_m, X_m) \rangle$, 
or simply $f = \langle (f_i, X_i) \rangle$. 
We assume that such partition into subintervals are maximal in a sense that for any $i$ and $i+1$ $f_i\ne f_{i+1}$. 
We call each pair $(f_i, X_i)$ a {\it piece} of $f$, 
and an endpoint of the closure of $X_i$ a {\it breakpoint} of $f$. 
A piecewise polynomial function $f = \langle (f_i, X_i) \rangle$ 
is called a {\it piecewise polynomial function of degree at most two}
if and only if each $f_i$ is quadratic or linear. 
We confirm the following property about the sum of piecewise polynomial functions. 
\begin{proposition}\label{pro:sum_piecewise}
Let $m$ and $m'$ be positive integers, and 
$f, g: X (\subset {\mathbb R}) \rightarrow {\mathbb R}$ 
be piecewise polynomial functions of degree at most two with $m$ and $m'$ pieces, respectively. 
Then, a function $h = f + g$ is a piecewise polynomial function of degree at most two 
with at most $m + m'$ pieces. 
Moreover, given $f = \langle (f_i, X_i) \rangle$ and 
$g = \langle (g_j, X'_j) \rangle$, 
we can obtain $h = f + g = \langle (h_j, X''_j) \rangle$ in $O( m + m' )$ time.
\end{proposition}

Let ${\cal F} = \{f_1(y), \ldots, f_m(y)\}$ be a family of $m$ polynomial functions 
where $f_i: Y_i(\subset \mathbb{R}) \rightarrow \mathbb{R}$
and $Y$ denote the union of $Y_i$, that is, $Y = \cup_{i=1}^{m}Y_i$.
An {\it upper envelope} ${\cal U}_{{\cal F}}(y)$ and a {\it lower envelope} ${\cal L}_{{\cal F}}(y)$ of ${\cal F}$ 
are functions from $Y$ to $\mathbb{R}$ defined as follows: 
\begin{eqnarray}\label{eq:envelope}
% {\cal U}_{{\cal F}}(y) = \max\{ f_i(y) \mid i = 1, \ldots, m \}, \:\:
% {\cal L}_{{\cal F}}(y) = \min\{ f_i(y) \mid i = 1, \ldots, m \},
{\cal U}_{{\cal F}}(y) = \max_{i = 1, \ldots, m} f_i(y), \:\:
{\cal L}_{{\cal F}}(y) = \min_{i = 1, \ldots, m} f_i(y),
\end{eqnarray}
where the maximum and the minimum are taken over those functions that are defined at $y$, respectively.
For an upper envelope ${\cal U}_{{\cal F}}(y)$ of ${\cal F}$, there exist
an integer sequence $U_{\cal F} = \langle u_1, \ldots, u_k \rangle$ and 
subintervals $I_1, \ldots, I_k$ of $Y$
such that 
${\cal U}_{{\cal F}}(y) = \langle (f_{u_1}(y), I_1), \ldots, (f_{u_k}(y), I_k) \rangle$ holds.
That is, an upper envelope ${\cal U}_{{\cal F}}(y)$ can be represented as a piecewise polynomial function.
We call the above sequence $U_{\cal F}$ 
the {\it upper-envelope sequence} of ${\cal U}_{\cal F}(y)$. 

In our algorithm, we compute the upper/lower envelopes of partially defined, univariate polynomial functions. 
The following result is useful for this operation. 
\begin{theorem}[\cite{hart1986nonlinearity,Hershberger89,agarwal1989sharp}]\label{thm:envelope_theorem}
Let ${\cal F}$ be a family of $n$ partially defined, polynomial functions of degree at most two. 
Then, ${\cal U}_{{\cal F}}$ and ${\cal L}_{{\cal F}}$ consist of $O( n2^\alpha(n) ) $ pieces
and one can obtain them in time $O(n \alpha(n) \log n )$, where $\alpha(n)$ is the inverse Ackermann function. 
Moreover, if ${\cal F}$ a family of $n$ line segments, 
then ${\cal U}_{{\cal F}}$ and ${\cal L}_{{\cal F}}$ consist of $O(n)$ pieces
and one can obtain them in time $O(n \log n )$. 
\end{theorem}
Note that the number of pieces and the computation time for the upper/lower envelopes 
are involved with the maximum length of Davenport--Schinzel sequences. 
See~\cite{hart1986nonlinearity} for the details. 
%\begin{remark}
For a family ${\cal F}$ of functions, if we say that we {\it obtain} envelopes ${\cal U}_{{\cal F}}(y)$ or ${\cal L}_{{\cal F}}(y)$, 
then we obtain the information of all pieces $(f_{u_i}(y), I_i)$. 
%\end{remark}

%\section{Property of the Inverse Ackermann Function}\label{sec:inverse_Ackermann}
\subsection{Property of the Inverse Ackermann Function}\label{sec:inverse_Ackermann}

The Ackerman function is defined as follows:
\begin{eqnarray*}
A(n, m) = \left\{
 \begin{array}{ll}
  m + 1 & \text{ if } n = 0, \\
  A(n-1, 1) & \text{ if } n > 0, m=0, \\
  A(n-1, A(n, m - 1)) & \text{ otherwise.}
 \end{array}
\right.
\end{eqnarray*}
The inverse Ackermann function $\alpha(n)$ is defined as
\begin{eqnarray*}
\alpha(n) = \min \{ k \in {\mathbb N}_{0} \mid n \leq A(k, k) \}. 
\end{eqnarray*}
We show that the following inequality:
\begin{property}\label{prop:inverse_ackermann}
\begin{eqnarray}\label{eq:inverse_ackermann}
%\alpha(n^2) \leq \alpha(n) + 1
\alpha(n^3) \leq \alpha(n) + 1
\end{eqnarray}
holds for any positive integer $n$ with $n \geq 8$. 
\end{property}
Let us suppose that $n \geq 8$ and $k = \alpha(n) \geq 3$ since $A(2, 2) = 7$ holds. 
The inequality~\eqref{eq:inverse_ackermann} is equivalent to the inequality
\[
%n^2 \leq A(k+1, k+1).
n^3 \leq A(k+1, k+1).
\]
By the definition and the monotonicity of the Ackermann function, we have 
\begin{eqnarray*}
A(k+1, k+1) = A\left( k, A(k+1, k) \right) > A\left( k, A(k, k)\right) \geq 2^{A(k, k)} \geq \bigl( A(k, k) \bigr)^3.
%\geq \bigl( A(k, k) \bigr)^2.
\end{eqnarray*}
The last two inequalities are led by the fact that 
\[
A\left( k, m\right) \geq A\left( 3, m\right) = 2^{m+3} - 3 \geq m^3
\]
holds for any $k \geq 3$ and any positive integer $m$. 
Since $n \leq A(k, k)$ holds, 
we have $n^3 \leq \bigl( A(k, k) \bigr)^3 \leq A(k+1, k+1)$. 
%we have $n^2 \leq \bigl( A(k, k) \bigr)^2 \leq A(k+1, k+1)$. 
Thus, the proof completes.

%\section{Key Idea for Solving the Problem in Polynomial Time}\label{sec:keyidea}
%\section{Key Idea and Properties}\label{sec:keyidea}
\section{Algorithms}\label{sec:keyidea}

The main task of the algorithm is to compute the following $O(n)$ values, 
$MR(v)$ for all $v \in V$ and 
$\min\{ MR(x) \mid x \in e\}$ for all $e \in E$. 
Once we compute these values, we immediately obtain the solution of the problem
by choosing the minimum one among them in $O(n)$ time. 

Let us focus on computing 
$\min\{ MR(x) \mid x \in e\}$ for each $e \in E$. 
(Note that we can compute $MR(v)$ for $v \in V$ in a similar manner.)
Recall the definition of the maximum regret for $x$, 
\mbox{$MR(x) = \max\{ R(x,t) \mid t \in T \}$}. 
A main difficulty lies in evaluating $R(x,t)$ over $t \in T$ 
even for a fixed $x$ since we need treat an infinite set $T$. 
Furthermore, we are also required to find an optimal location among an infinite set $e$. 
To tackle with this issue, our key idea is to partition the problem into a polynomial number of subproblems as follows:
We partition interval $T$ into a polynomial number of subintervals $T_1, \ldots, T_{m}$ 
so that 
$R(x,t)$ is represented as a (single) polynomial function in $x$ and $t$ 
on $\{x \in e\} \times T_j$ for each $j = 1, \ldots, m$. 
For each $T_j$, 
we compute  
the maximum regret for $x \in e$ over $T_j$ 
denoted by $G_j(x) =\max \{R(x,t) \mid t \in T_j\}$. 
%{\color{red} An explicit form of $G_j(x)$ will be given in Sec.~\ref{app:proof_G^e}.}
An explicit form of $G_j(x)$ will be given in Sec.~\ref{sec:prop_Phi}. 
We then obtain $MR(x)$ for $x \in e $ as the upper envelope of functions $G_1(x), \ldots, G_{m}(x)$
and find the minimum value of $MR(x)$ for $x \in e$ by elementary calculation. 
%In a similar manner, we compute values $MR(v)$ for $v \in V$. 
%to compute $MR(v)$ for each $v \in V$, we partition interval $T$ into a polynomial number of subintervals $T^v_1, \ldots, T^v_{N_v}$ so that 
%for $t \in T^v_j$,
%$R(v,t) = \Phi(v,t)-{\rm Opt}(t)$ is represented as a single polynomial in $t$.
%We then easily compute for $j = 1, \ldots, N_v$
%\begin{eqnarray*}
%G^v_j=\max \{R(v,t) = \Phi(v,t)-{\rm Opt}(t) \mid t \in T^v_j\},
%\end{eqnarray*}
%and obtain $MR(v)$ as the maximum of  $G^v_1, \ldots, G^v_{N_v}$ in $O(N_v)$ time.

In the rest of the paper, we mainly 
show that
for each $e$ or $v$,
there exists a partition of $T$ with a polynomial number of subintervals
such that the regret $R(x,t)$ is a polynomial function of degree at most two on each subinterval.

%\section{Key Lemmas}\label{sec:prop_Phi}
\subsection{Key Lemmas}\label{sec:prop_Phi} 

To understand $R(x,t)$, we observe function $\Phi(x,t)$. 
% Before showing the properties of $\Phi(x,t)$, 
We give some other notations. 
Let $f^{e_i, j}_{{\rm R}} (t, z)$ and $f^{e_i, j}_{{\rm L}} (t, z)$ denote functions obtained by removing terms containing $x$ from formulae~\eqref{eq:ct_right_1} and \eqref{eq:ct_left}.
Formally, for $1 \leq i < j \leq n$, let function $f^{e_i, j}_{{\rm R}} (t, z)$ be defined on $t \in T$ and $z \in (W_{i+1,j-1}(t), W_{i+1,n}(t)]$ as 
\begin{eqnarray}\label{eq:ct_right_sub}
f^{e_i, j}_{{\rm R}} (t, z) = \tau v_j + \frac{z - W_{i+1,j-1}(t) }{C_{i, j}}, 
\end{eqnarray}
and for $1 \leq j < i \leq n$, let function $f^{e_i, j}_{{\rm L}} (t, z)$ be defined on $t \in T$ and $z \in (W_{j+1,i}(t), W_{1,i}(t)]$ as
\begin{eqnarray}\label{eq:ct_left_sub}
f^{e_i, j}_{{\rm L}} (t, z) = - \tau v_j + \frac{z - W_{j+1,i}(t) }{C_{j,i}}.
\end{eqnarray}
In addition, let $F^{e_i}_{\rm L}(t)$ and $F^{e_i}_{\rm R}(t)$ denote univariate functions defined as \begin{eqnarray}\label{eq:F_e}
%F^{e_i}(t) = \int_{0}^{W_{1,i}(t)} f^{e_i}_{{\rm L}} (t, z) dz + \int_{0}^{W_{i+1,n}(t)} f^{e_i}_{{\rm R}} (t, z) dz.
F^{e_i}_{\rm L}(t) = \int_{0}^{W_{1,i}(t)} f^{e_i}_{{\rm L}} (t, z) dz, \quad
F^{e_i}_{\rm R}(t) =  \int_{0}^{W_{i+1,n}(t)} f^{e_i}_{{\rm R}} (t, z) dz, 
\end{eqnarray}
where $f^{e_i}_{{\rm L}} (t, z)$ and $f^{e_i}_{{\rm R}} (t, z)$ denote functions defined as
\[
%f^{e_i}_{{\rm L}} (t, z) = \max \left\{  f^{e_i, j}_{{\rm L}} (t, z) \mathrel{}\middle|\mathrel{} 1 \leq j \leq i \right\}, 
%f^{e_i}_{{\rm R}} (t, z) = \max \left\{  f^{e_i, j}_{{\rm R}} (t, z) \mathrel{}\middle|\mathrel{} i + 1 \leq j \leq n \right\}.
f^{e_i}_{{\rm L}} (t, z) = \max_{1 \leq j \leq i} \left\{  f^{e_i, j}_{{\rm L}} (t, z) \right\}, \
f^{e_i}_{{\rm R}} (t, z) = \max_{i + 1 \leq j \leq n } \left\{  f^{e_i, j}_{{\rm R}} (t, z) \right\}.
\]
Recall the definition of the aggregate evacuation time $\Phi(x,t)$ shown in~\eqref{eq:at}. 
We observe that
for $x \in e_i$, 
$\Phi(x,t)$ can be represented as 
\begin{eqnarray}\label{eq:at_with_F}
\Phi( x, t ) & = &  \bigl(W_{1,i}(t) - W_{i+1,n}(t) \bigr) \tau x + \!\! \int_{0}^{W_{1,i}(t)} \!\!\!\!\!\!\!\! f^{e_i}_{{\rm L}} (t, z) dz + \!\! \int_{0}^{W_{i+1,n}(t)} \!\!\!\!\!\!\!\! f^{e_i}_{{\rm R}} (t, z) dz \notag \\
& = & \bigl(W_{1,i}(t) - W_{i+1,n}(t) \bigr) \tau x + F^{e_i}_{\rm L}(t) + F^{e_i}_{\rm R}(t).
\end{eqnarray}
In a similar manner, 
by the definition of \eqref{eq:at_v} and formula~\eqref{eq:F_e}, 
we have \begin{eqnarray}\label{eq:at_v_with_F}
%\Phi( v_i, t ) =  \bigl( W_{1,i-1}(t) - W_{i+1,n}(t) \bigr) \tau v_i + F^{v_i}(t). 
\Phi( v_i, t ) =  \bigl( W_{1,i-1}(t) - W_{i+1,n}(t) \bigr) \tau v_i + F^{e_{i-1}}_{\rm L}(t) + F^{e_i}_{\rm R}(t). 
\end{eqnarray}

Let us focus on function $F^e_{\rm R}(t)$. 
As $t$ increases, while the upper-envelope sequence of $f_{\rm R}^e(t,z)$ w.r.t. $z$ remains the same, 
function $F^e_{\rm R}(t)$ is represented as the same polynomial,
whose degree is at most two by formulae~\eqref{eq:ct_right_sub}, \eqref{eq:ct_left_sub} and \eqref{eq:F_e}.
In other words, 
a breakpoint of $F^e_{\rm R}(t)$ corresponds to the value $t$ such that 
the upper-envelope sequence of $f_{\rm R}^e(t,z)$ w.r.t. $z$ changes. 
We notice that such a change happens only when three functions $f^{e, h}_{{\rm R}} (t, z)$,
$f^{e, i}_{{\rm R}} (t, z)$ and $f^{e, j}_{{\rm R}} (t, z)$ 
intersect each other, which can happen at most once. 
This implies that $F^e_{\rm R}(t)$ consists of $O(n^3)$ breakpoints, that is, 
it is a piecewise polynomial function of degree at most two with $O(n^3)$ pieces.
The following lemma shows that the number of pieces is actually $O(n^2)$. 
%{\color{red} See Appendix~\ref{app:lem_func_form} for details of the proof.}

\begin{figure}[t]%[htbp]
  \begin{center}
    \includegraphics[width=8cm,pagebox=cropbox,clip]{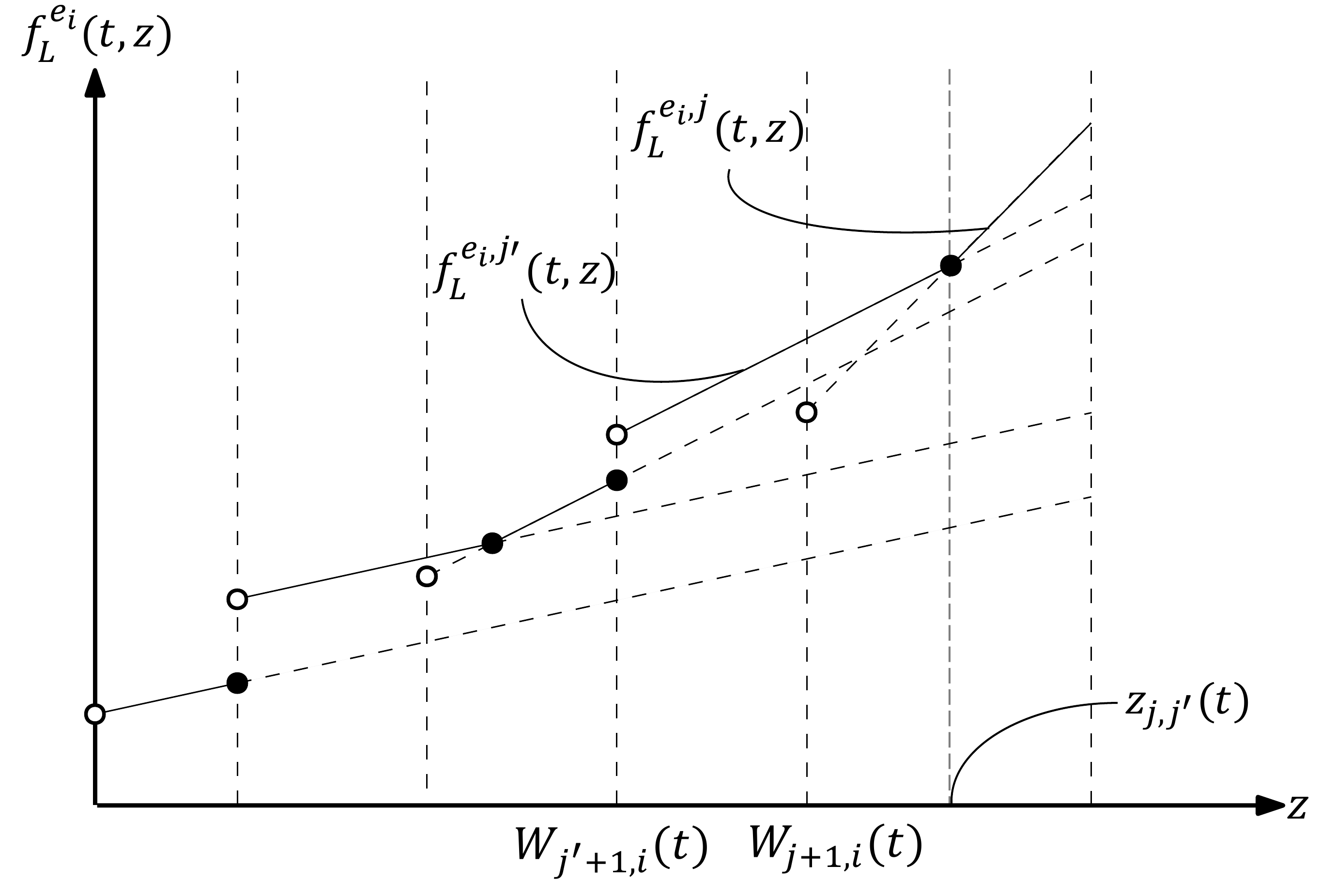}
    \caption{Graph of $f^{e_i}_{{\rm L}} (t, z)$ w.r.t. $z$ with a focus on $f^{e_i,j}_{{\rm L}} (t, z)$, $f^{e_i,j'}_{{\rm L}} (t, z)$ and $z_{j,j'}(t)$ with $j < j'$.}
    \label{fig:lemma7_1}
  \end{center}
\end{figure}

\begin{lemma}\label{lem:func_Phi_merged}
%For each $e \in E$ and each $v \in V$, 
%$F^{e}(t)$ and $F^{v}(t)$ are piecewise polynomial functions 
For each $e \in E$, 
$F^{e}_{\rm L}(t)$ and $F^{e}_{\rm R}(t)$ are piecewise polynomial functions 
of degree at most two with $O(n^2)$ pieces, and can be computed in $O(n^3 \log n)$ time.
Especially, when all the edge capacities are uniform, the numbers of pieces of them are $O(n)$,
and can be computed in $O(n^2 \log n)$ time.
\end{lemma}

\begin{proof}
%\section{Proof of Lemma~\ref{lem:func_Phi_form}}\label{app:lem_func_form}
%\section{Proof of Lemma~\ref{lem:func_Phi_merged}}\label{app:lem_func_form}
%\begin{proof}
Let us suppose that $x$ is on $e_i \in E$. 
We will show that $F^{e_i}_{\rm L}(t)$ is a piecewise polynomial functions of degree at most two with $O(n^2)$ pieces, and when all the edge capacities are uniform, the numbers of pieces of them are $O(n)$. 
For $F^{e_i}_{\rm R}(t)$, we can see the same properties in a symmetric manner. 

To show the above statement, we give the following properties of $F^{e_i}_{\rm L}(t)$: 
\begin{itemize}
\item[(i)] The number of pieces of $F^{e_i}_{\rm L}(t)$ is $O(n^2)$. 
Especially, when all the edge capacities are uniform, it is $O(n)$.
\item[(ii)] Function of each piece of $F^{e_i}_{\rm L}(t)$ is a quadratic function in $t$. 
\end{itemize}

% Condition~{(i)}
First, we show that property~{(i)} holds. 
Recall that
$f^{e_i}_{{\rm L}} (t, z)$ is 
the upper envelope of $\left\{f^{e_i, j}_{{\rm L}} (t, z) \mathrel{}\middle|\mathrel{} 1 \leq j \leq i \right\}$ w.r.t. 
$z \in (0, W_{1,i}(t)]$.
% \[
% f^{e_i}_{{\rm L}} (t, z) = \max \left\{  f^{e_i, j}_{{\rm L}} (t, z) \mathrel{}\middle|\mathrel{} 1 \leq j \leq i \right\}.
% \]
For integers $j,j'$ with $1 \leq j < j' \leq i$,
let $z_{j,j'}(t)$ be a function in $t$ defined as follows:
If $C_{j,i}=C_{j',i}$,
\begin{eqnarray}\label{eq:z(t)_uni}
z_{j,j'}(t) = \left\{ 
\begin{array}{ll}
W_{j+1,i}(t) & \text{ if } f^{e_i, j'}_{{\rm L}} (t, W_{j+1,i}(t)) \le f^{e_i, j}_{{\rm L}} (t, W_{j+1,i}(t)), \\
W_{1,i}(t) & \text{ if } f^{e_i, j'}_{{\rm L}} (t, W_{j+1,i}(t)) > f^{e_i, j}_{{\rm L}} (t, W_{j+1,i}(t)), \\
\end{array}
\right.
\end{eqnarray}
otherwise, i.e., if $C_{j,i}<C_{j',i}$,
\begin{eqnarray}\label{eq:z(t)_gen}
z_{j,j'}(t) = \left\{ 
\begin{array}{ll}
W_{j+1,i}(t) & \text{ if } f^{e_i, j'}_{{\rm L}} (t, W_{j+1,i}(t)) \le f^{e_i, j}_{{\rm L}} (t, W_{j+1,i}(t)), \\
\min\{z^*_{j,j'}(t),W_{1,i}(t)\} & \text{ if } f^{e_i, j'}_{{\rm L}} (t, W_{j+1,i}(t)) > f^{e_i, j}_{{\rm L}} (t, W_{j+1,i}(t)), \\
\end{array}
\right. \notag \\
\end{eqnarray}
%\begin{eqnarray}
%z_{j,j'}(t) = \max \left\{ W_{j+1,i}(t), z^*_{j,j'}(t) \right\},
%\end{eqnarray}
where $z^*_{j,j'}(t)$ is the solution for $z$ of equation $f^{e_i, j}_{{\rm L}} (t, z) = f^{e_i, j'}_{{\rm L}} (t, z)$ (See Figure~\ref{fig:lemma7_1}).
% \begin{figure}[t]%[htbp]
%   \begin{center}
%     \includegraphics[width=8cm,pagebox=cropbox,clip]{fig_lemma1.pdf}
%     \caption{Graph of $f^{e_i}_{{\rm L}} (t, z)$ w.r.t. $z$ with a focus on $f^{e_i,j}_{{\rm L}} (t, z)$, $f^{e_i,j'}_{{\rm L}} (t, z)$ and $z_{j,j'}(t)$ with $j < j'$.}
%     \label{fig:lemma7_1}
%   \end{center}
% \end{figure}
Note that
inequality conditions in \eqref{eq:z(t)_uni} and \eqref{eq:z(t)_gen}
can be solved for $t$ as [$t \le t_{j,j'}$; $t > t_{j,j'}$] or [$t \ge t_{j,j'}$; $t < t_{j,j'}$],
where $t_{j,j'}$ is the solution for $t$ of equation $f^{e_i, j'}_{{\rm L}} (t, W_{j+1,i}(t)) = f^{e_i, j}_{{\rm L}} (t, W_{j+1,i}(t))$.

For any $t \in T$, let $U(t)$ be the upper-envelope sequence of $f^{e_i}_{{\rm L}}(t, z)$ in $z$.
Now, for some $t'$, suppose
$U(t') = \langle u_1, \ldots, u_k \rangle$.
Note that $u_1 > \ldots > u_k$ holds.
Then, 
%for an integer $h$ with $1 \leq h \leq k-1$,
the $h$-th smallest breakpoint of $f^{e_i}_{{\rm L}}(t', z)$ for $z \in (0, W_i(t'))$ is $z_{u_{h+1},u_h}(t')$.
We notice that
if $t$ changes from $t'$ in condition that 
$U(t)$ and formula of each $z_{u_{h+1},u_h}(t)$ remain the same,  
formula of $F^{e_i}_{\rm L}(t)$ also remains the same.
In other words, a breakpoint of $F^{e_i}_{\rm L}(t)$ corresponds to when $U(t)$ or formula of some $z_{u_{h+1},u_h}(t)$ changes.
Consider the following three cases.

\noindent
[Case 1]: 
When $t$ comes to some value $t_1$, formula of some $z_{u_{h+1},u_h}(t)$ changes.

\noindent
[Case 2]:
Just after $t$ comes to some value $t_2$,
i.e., when $t=t_2+\epsilon$, $f^{e_i, j}_{{\rm L}} (t_2+\epsilon, z)$ of some $j$ appears as a part of $f^{e_i}_{{\rm L}} (t_2+\epsilon, z)$
so that $U(t_2+\epsilon) = \langle u_1, \ldots, u_h, j, u_{h+1}, \ldots, u_k \rangle$.

\noindent
[Case 3]:
When $t$ comes to some value $t_3$, $f^{e_i, u_h}_{{\rm L}} (t_3, z)$ disappears from $f^{e_i}_{{\rm L}} (t_3, z)$,
i.e., $U(t_3) = \langle u_1, \ldots, u_{h-1}, u_{h+1}, \ldots, u_k \rangle$.

We observe that each of the above $t_1,t_2$ and $t_3$ is a breakpoint of
\begin{eqnarray}\label{eq:z_u_h+1}
z_{u_{h+1}}(t) := \max_{j'} \left\{ z_{u_{h+1},j'}(t) \mid u_{h+1}+1 \le j' \le i \right\}.
\end{eqnarray}
Therefore, the number of breakpoints of $F^{e_i}_{\rm L}(t)$ is at most the number of all the breakpoints of
\begin{eqnarray}\label{eq:z_j}
z_j(t) := \max_{j'} \left\{ z_{j,j'}(t) \mid j+1 \le j' \le i \right\} \quad \forall j \ \ \text{with} \ \ 1 \le j \le i-1,
\end{eqnarray}
which is $O(n^2)$ since the number of breakpoints of $z_j(t)$ for any $j$ is $O(n)$.
Especially, when all the edge capacities are uniform, according to \eqref{eq:z(t)_uni} and \eqref{eq:z_j},
we have
\begin{eqnarray}\label{eq:z(t)_uni_2}
z_{j}(t) = \left\{ 
\begin{array}{ll}
W_{j+1,i}(t) & \text{ if } f^{e_i, j'}_{{\rm L}} (t, W_{j+1,i}(t)) \le f^{e_i, j}_{{\rm L}} (t, W_{j+1,i}(t)) \\
 & \qquad \qquad \qquad \qquad \qquad \qquad \forall j' \text{ with } j+1 \le j' \le i, \\
W_{1,i}(t) & \text{ otherwise}. \\
\end{array}
\right.
\end{eqnarray}
Thus, the number of breakpoints of $z_j(t)$ for any $j$ is at most two,
which means that the number of breakpoints of $F^{e_i}_{\rm L}(t)$ is $O(n)$.
This completes the proof for property~{(i)}. 

%%% Condition (ii) 
We next show that property~{(ii)} holds. 
Let $(H(t),T')$ be a piece of $F^{e_i}_{\rm L}(t)$
such that for $t \in T'$, $U(t) = \langle u_1, \ldots, u_k \rangle$ and 
formula of each $z_{u_{h+1},u_h}(t)$ remains the same.
We then have
\begin{eqnarray}\label{eq:H(t)}
H(t) = \sum_{h=1}^{k} \int_{z_{u_h,u_{h-1}}(t)}^{z_{u_{h+1},u_h}(t)} f^{e_i,u_h}_{{\rm L}} (t, z) dz,
\end{eqnarray}
where $z_{u_1,u_0}(t) = 0$ and $z_{u_{k+1},u_{k}}(t) = W_{1,i}(t)$.
%where $z_{u_0,u_1}(t) = 0$ and $z_{u_{k.},u_{k+1}}(t) = W_{1,i}(t)$. 
By formula~\eqref{eq:ct_left_sub}, it holds for $h$ with $1 \le h \le k$ that
\begin{eqnarray}\label{eq:H(t)_2}
f^{e_i,u_h}_{{\rm L}} (t, z) = \rho^h_1 z + \rho^h_2 t + \rho^h_3,
\end{eqnarray}
where $\rho^h_1, \rho^h_2, \rho^h_3$ are some constants.
Substituting \eqref{eq:H(t)_2} into \eqref{eq:H(t)}, we obtain
\begin{eqnarray}\label{eq:H(t)_3}
H(t) = \sum_{h=1}^{k} \left\{\frac{\rho^h_1}{2} \left(z_{u_{h+1},u_h}(t)^2-z_{u_h,u_{h-1}}(t)^2 \right)+(\rho^h_2t+\rho^h_3) \left(z_{u_{h+1},u_h}(t)-z_{u_h,u_{h-1}}(t) \right)\right\}, %\notag \\
\end{eqnarray}
which is quadratic in $t$ since $z_{u_{h+1},u_h}(t)$ and $z_{u_h,u_{h-1}}(t)$ are linear in $t$.
This completes the proof for property~{(ii)}. 

The rest of the proof is to give an algorithm for obtaining 
$F^{e_i}_{\rm L}(t) = \langle ( F^{e_i}_{{\rm L}, h}(t), T_{{\rm L}, h} ) \rangle$. 
Let $m$ denote the number of pieces of $F^{e_i}_{\rm L}(t)$. 
Note that one can obtain $F^{e_i}_{\rm R}(t)$ in a similar manner. 
Our algorithm consists of the following two steps:

\noindent
{\bf Step.~1: Obtain breakpoints of $F^{e_i}_{\rm L}(t)$.} 
For each $j$ with $1 \leq j \leq i$, 
we find all the breakpoints of an upper envelope $z_{j}(t)$. 
Since $z_{j}(t)$ consists of partially defined linear functions, 
we can apply Theorem~\ref{thm:envelope_theorem}. 
Therefore, this operation requires $O(n \log n)$ time for each $j$. 
Thus, Step~1 requires $O(n^2 \log n)$ time.

\noindent
{\bf Step.~2: Obtain polynomial functions of all the pieces of $F^{e_i}_{\rm L}(t)$.}
Let $T'$ be an interval of a piece of $F^{e_i}_{\rm L}(t)$. 
Picking up a parameter $t \in T'$, we obtain 
an upper-envelope sequence $U(t) = \langle u_1, \ldots, u_k \rangle$ 
in $O(n \log n)$ time by Theorem~\ref{thm:envelope_theorem}. 
We obtain polynomial of degree at most two in $O(n)$ time 
by evaluating formula~\eqref{eq:H(t)_3}. 
Since the number of pieces of $F^{e_i}_{\rm L}(t)$ is at most $m$, 
Step~2 requires $O(n m \log n)$ time. 

In total, our algorithm requires 
$O(n (n + m) \log n)$ time, 
%$O(n N_{F} \log n)$ time 
which completes the proof of our main theorem 
because $m = O(n^2)$, and for the case with uniform edge capacity, $m = O(n)$.
%\qed
\end{proof}

%Let $N_{F}$ denote the maximum number of pieces of $F^{e}(t)$ and $F^{v}(t)$ over $e \in E$ and $v \in V$. 
Let $N_{F}$ denote the maximum number of pieces of $F^{e}_{\rm L}(t)$ and $F^{e}_{\rm R}(t)$ over $e \in E$. 
Then we have $N_{F} = O(n^2)$, and for the case with uniform edge capacity, $N_{F} = O(n)$. 
Next, we consider ${\rm Opt}(t) = \min\{ \Phi(x,t) \mid x \in V\}$, which is the lower envelope of a family of $n$ functions $\Phi(v_i, t)$ in $t$.
%$\Phi(v_i, t) =  \bigl( W_{i-1}(t) + W_i(t) - W_n(t) \bigr) \tau v_i + F^{v_i}(t)$.
Theorem~\ref{thm:envelope_theorem} and Lemma~\ref{lem:func_Phi_merged}
imply the following lemma. 
%{\color{red} See Appendix~\ref{app:opt} for the proof.}

\begin{lemma}\label{lem:opt} 
${\rm Opt}(t)$ is a piecewise polynomial function
of degree at most two with $O(n N_{F}2^{\alpha(n)})$ pieces,
and can be obtained in $O(n N_{F} \alpha(n) \log n)$ time
%if functions $F^{v} for all $v \in V$ are available.
if functions $F^{e}_{\rm L}(t)$ and $F^{e}_{\rm R}(t)$ for all $e \in E$ are available.
\end{lemma}
%\section{Proof of Lemma~\ref{lem:opt}}\label{app:opt}
\begin{proof}
Formula~\eqref{eq:opt} implies that ${\rm Opt}(t)$ is the lower envelope of 
a family of $n$ functions \mbox{$\{\Phi(v,t) \mid v \in V \}$}.
Recall that for $v_i \in V$, we have
\begin{eqnarray*}
\Phi(v_i, t) =  \bigl( W_{1,i-1}(t) - W_{i+1,n}(t) \bigr) \tau v_i + F^{e_{i-1}}_{\rm L}(t) + F^{e_i}_{\rm R}(t). 
\end{eqnarray*}
By Lemma~\ref{lem:func_Phi_merged}, since $F^{e}_{\rm L}(t)$ and $F^{e}_{\rm R}(t)$ for any $e \in E$ is a piecewise polynomial function  
of degree at most two with at most $N_{F}$ pieces, 
${\rm Opt}(t)$ is the lower envelope of 
$O(n N_{F})$ partially defined polynomial functions of degree at most two. 

Theorem~\ref{thm:envelope_theorem} implies that 
${\rm Opt}(t)$ is a piecewise polynomial function of degree at most two 
with at most $O(n N_{F} 2^{\alpha(n N_{F})}) = O(n N_{F} 2^{\alpha(n)})$ pieces and 
can be obtained in $O(n N_{F} \alpha(n N_{F}) \log (n N_{F}))) = O(n N_{F} \alpha(n) \log n)$ time, 
where we used the facts that $N_{F} = O(n^3)$ and 
Property~\ref{prop:inverse_ackermann}.
It completes the proof. 
%\qed
\end{proof}

Let $N_{\rm Opt}$ denote the number of pieces of ${\rm Opt}(t)$. 
Then we have $N_{\rm Opt} = O(n N_{F}2^{\alpha(n)})$.

Let us consider $R(x, t)$ in the case that sink $x$ is on an edge $e_i \in E$.
Substituting formula~\eqref{eq:at_with_F} for~\eqref{def:regret}, we have 
\[
%R(x, t) = \Phi(x,t) - {\rm Opt}(t) = \bigl( W_{1,i}(t) - W_{i+1,n}(t) \bigr) \tau x + F^{e_i}(t) - {\rm Opt}(t). 
R(x, t) = \Phi(x,t) - {\rm Opt}(t) = \bigl( W_{1,i}(t) - W_{i+1,n}(t) \bigr) \tau x 
+ F^{e_i}_{\rm L}(t)+ F^{e_i}_{\rm R}(t) - {\rm Opt}(t). 
\]
By Proposition~\ref{pro:sum_piecewise}, 
%$F^{e_i}(t) - {\rm Opt}(t)$ 
$F^{e_i}_{\rm L}(t)+ F^{e_i}_{\rm R}(t) - {\rm Opt}(t)$
is a piecewise polynomial function of degree at most two
with at most $2 N_{F} + N_{\rm Opt} = O(N_{\rm Opt})$ pieces.
Let $N_{e_i}$ be the number of pieces of 
%$F^{e_i}(t) - {\rm Opt}(t)$ 
$F^{e_i}_{\rm L}(t)+ F^{e_i}_{\rm R}(t) - {\rm Opt}(t)$ and 
$T^{e_i}_j$ be the interval of the $j$-th piece (from the left) of 
%$F^{e_i}(t) - {\rm Opt}(t)$. 
$F^{e_i}_{\rm L}(t)+ F^{e_i}_{\rm R}(t) - {\rm Opt}(t)$. 
Thus, $R(x,t)$ is represented as a (single) polynomial function in $x$ and $t$ on $\{x \in e\} \times T_j$ for each $T_j$. 
For each integer $j$ with  $1 \leq j \leq N_{e_i}$, let $G^{e_i}_{j}(x)$ be a function defined as
\begin{eqnarray}
G^{e_i}_{j}(x) = \max \{ R(x, t) \mid t \in T^{e_i}_j\}.
\end{eqnarray}
%where $T^{e_i}_j$ is the interval of the $j$-th piece (from the left) of $F^{e_i}(t) - {\rm Opt}(t)$. 
We then have the following lemma. 
%{\color{red} See Appendix~\ref{app:proof_G^e} for the proof.} 
\begin{lemma}\label{lem:G^e_function} 
For each $e_i \in E$ and $j$ with $1 \le j \le N_{e_i}$, $G^{e_i}_{j}(x)$ is a piecewise polynomial function of degree at most two with at most three pieces, and can be obtained in constant time
if functions $F^{e_i}_{\rm L}(t)$, $F^{e_i}_{\rm R}(t)$ and ${\rm Opt}(t)$ are available.
%if functions $F^{e_i}(t)$ and ${\rm Opt}(t)$ are available.
\end{lemma}
\begin{proof}
%\section{Proof of Lemma~\ref{lem:G^e_function}}\label{app:proof_G^e}
Recall that we have 
\[
R(x, t) = \Phi(x,t) - {\rm Opt}(t) = \bigl( W_{1,i}(t) - W_{i+1,n}(t) \bigr) \tau x 
+ F^{e_i}_{\rm L}(t)+ F^{e_i}_{\rm R}(t) - {\rm Opt}(t). 
\]
Because $W_{1,i}(t)$, $W_{i+1,n}(t)$ are linear in $t$ and
$F^{e_i}_{\rm L}(t)+ F^{e_i}_{\rm R}(t) - {\rm Opt}(t)$ are 
polynomial function of degree at most two on $t \in T_j$,
we can represent $R(x, t)$ on $\{x \in e_i\} \times T_j$ 
with five real constants $\beta_k$ ($k = 1, \ldots, 5$) as
\[
R(x, t) = \beta_1 t^2 + \beta_2 x t + \beta_3 t + \beta_4 x + \beta_5.
\]

Let us obverse the explicit form of the maximum value $G_{j}(x)$ of $R(x,t)$ for $x \in e_i$ over $t \in T_j$. 
Let $T_j = [t_{j-1}, t_{j}]$. 
If $\beta_1 \geq 0$, 
$R(x, t)$ takes the maximum value when $t = t_{j-1}$ or $t=t_{j}$ for any $x$. 
Thus, we have $G_{j}(x) = \max\{ R(x, t_{j-1}), R(x, t_{j})\}$ that is 
a piecewise linear function in $x$ with at most two pieces, 
since both $R(x, t_{j-1})$ and $R(x, t_{j})$ are linear in $x$. 
Let us consider the other case that $\beta_1 < 0$ holds. 
When the axis of symmetry of $R(x,t)$ w.r.t. $t$, 
i.e., $t = -(\beta_2x +\beta_3)/(2\beta_1)$, is contained in $T_j$,
it holds $G_{j}(x) = R(x, -(\beta_2x +\beta_3)/(2\beta_1))$.
We thus have
\begin{eqnarray}\label{eq:mrj}
G_{j}(x) = \left\{ 
\begin{array}{ll}
R(x, t_{j-1})& \text{ if } -\frac{\beta_2x + \beta_3}{2\beta_1} < t_{j-1}, \\
R\left( x, -\frac{\beta_2x + \beta_3}{2\beta_1} \right) & \text{ if } t_{j-1} \leq -\frac{\beta_2x + \beta_3}{2\beta_1} \leq t_{j}, \\
R(x, t_{j})& \text{ if } -\frac{\beta_2x + \beta_3}{2\beta_1} > t_{j}.
\end{array}
\right.
\end{eqnarray}
Note that
inequality conditions in \eqref{eq:mrj}
can be solved for $x$ as [$x < x_1$; $x_1 \le x \le x_2$; $x > x_2$] or [$x > x_1$; $x_1 \ge x \ge x_2$; $x < x_2$],
where 
$x_1$ is the solution for $x$ of equation $-(\beta_2x + \beta_3)/(2\beta_1) = t_{j-1}$
and
$x_2$ is the solution for $x$ of equation $-(\beta_2x + \beta_3)/(2\beta_1) = t_{j}$ (See Figure~\ref{fig:mrj}).
Therefore, $G_{j}(x)$ is a piecewise polynomial function with at most three polynomial of degree at most two 
because $R\left( x, - \frac{\beta_2x + \beta_3}{2\beta_1} \right)$ is a quadratic function in $x$. %\qed
\end{proof}
\begin{figure}[t]%[htbp]
  \begin{center}
    \includegraphics[width=8cm,pagebox=cropbox,clip]{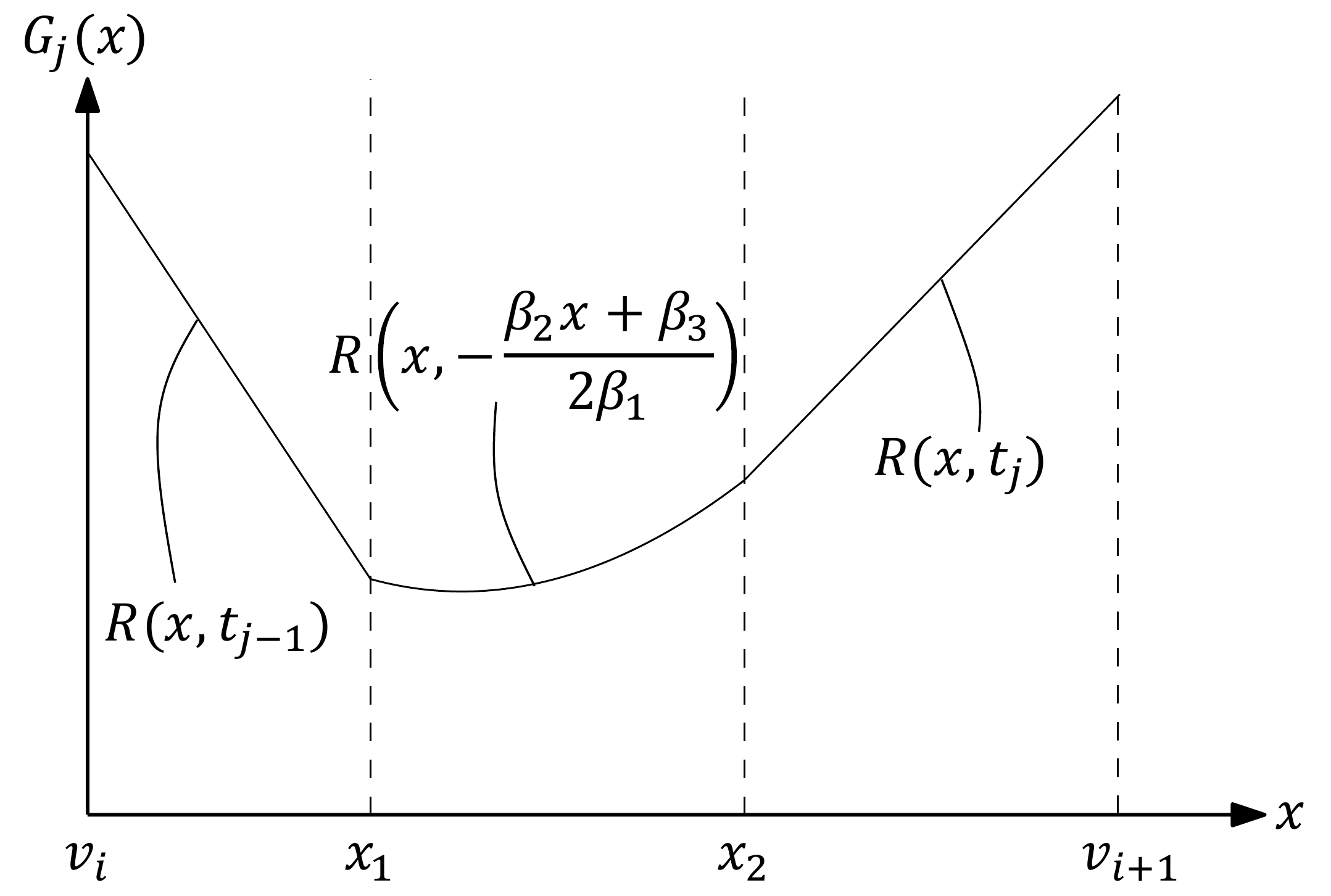}
    \caption{Graph of $G_{j}(x)$ in the case of $v_i < x_1 < x_2 < v_{i+1}$.}
    \label{fig:mrj}
  \end{center}
\end{figure}

Recalling the definition of $MR(x)$, it holds that for $x \in e$,
\[
MR(x) = \max \{R(x, t) \mid t \in T \} = 
\max \{G^{e}_{j}(x) \mid 1 \leq j \leq N_{e} \},
\]
that is, $MR(x)$ is the upper envelope of functions $G^{e}_{1}(x), \ldots, G^{e}_{N_{e}}(x)$. 
Applying Theorem~\ref{thm:envelope_theorem}, we have the following lemma. 
%{\color{red} See Appendix~\ref{app:proof_MR_e} for the proof.} 
\begin{lemma}\label{lem:MR_e_function} 
For each $e \in E$, there exists an algorithm that finds a location that 
minimizes $MR(x)$ under the restriction with $x \in e$
in $O(N_{\rm Opt} \alpha(n) \log n)$ time
%if functions $F^{e}(t)$ and ${\rm Opt}(t)$ are available.
if functions $F^{e}_{\rm L}(t)$, $F^{e}_{\rm R}(t)$ and ${\rm Opt}(t)$ are available.
\end{lemma}
%\section{Proof of Lemma~\ref{lem:MR_e_function}}\label{app:proof_MR_e}
\begin{proof}
We give how to find a location $x^{*, e}$ that minimizes $MR(x)$ over $x \in e$. 
Since functions $F^{e}_{\rm L}(t)$, $F^{e}_{\rm R}(t)$ and ${\rm Opt}(t)$ are available, 
we can apply Lemma~\ref{lem:G^e_function} and then 
compute the explicit forms of functions $G^{e}_{j}(x)$ 
for all $j$ with $1 \le j \le N_{e}$ which are obtained in $O(N_{e})=O(N_{\rm Opt})$ time. 
Function $MR(x)$ for $x \in e$ is the upper envelope of functions $G^{e}_{1}(x), \ldots, G^{e}_{N_{e}}(x)$.
Since $G^{e}_{j}(x)$ is a piecewise polynomial function of degree at most two with at most three pieces,
$MR(x)$ is the upper envelope of at most $3N_e = O(N_{\rm Opt})$ functions of degree at most two. 
Theorem~\ref{thm:envelope_theorem} implies that 
$MR(x)$ consists of $O(N_{\rm Opt} 2^{\alpha( N_{\rm Opt} )}) = O(N_{\rm Opt} 2^{\alpha(n)})$ pieces
and can be obtained in $O(N_{\rm Opt} \alpha( N_{\rm Opt} ) \log N_{\rm Opt}) = O(N_{\rm Opt} \alpha(n) \log n)$ time, 
%where we used the fact that $N_{\rm Opt} = O(n^3 2^{\alpha(n)})$ and the inverse Ackermann function satisfies 
where we used the fact that the inverse Ackermann function satisfies 
that $\alpha(n^3) \leq \alpha(n) + 1$ holds for any positive integer $n$ with $n \geq 8$. 
%(See Appendix~\ref{sec:inverse_Ackermann} for the fact of the inverse Ackermann function.)
%(See Appendix~\ref{sec:inverse_Ackermann} for the details of the fact.) 

For each piece, compute a point which minimizes $MR(x)$ in constant time, and among the obtained values, choose the minimum one as $x^{*, e}$. 
Summarizing the above argument, these operations take $O(N_{\rm Opt} \alpha(n) \log n)$ time. %\qed
\end{proof}

%{\color{red}
Note that the small modification for the algorithm of Lemma~\ref{lem:MR_e_function}
leads that we can also compute $MR(v)$ for all $v \in V$ 
in $O(N_{\rm Opt})$ time.
\begin{lemma}\label{lem:MR_v_function} 
For each $v \in V$, there exists an algorithm that 
computes $MR(v)$ in $O( N_{\rm Opt} )$ time
if functions $F^{e}_{\rm L}(t)$, $F^{e}_{\rm R}(t)$ and ${\rm Opt}(t)$ are available for all $e \in E$.
\end{lemma}
\begin{proof}
%\section{ Computing $MR(v)$}\label{app:MR_v_function}
We show that it takes to obtain $MR(v_i)$ in $O(N_{\rm Opt})$ time for each $v_i \in V$. 
%Hence, we obtain $MR(v)$ for all $v \in V$ in $O(n N_{\rm Opt})$. 
Substituting formula~\eqref{eq:at_v_with_F} for~\eqref{def:regret}, we have 
\[
R(v_i, t) = \bigl(W_{i-1}(t) + W_i(t) - W_n(t) \bigr) \tau v_i 
+ F^{e_{i-1}}_{\rm L}(t) + F^{e_i}_{\rm R}(t) - {\rm Opt}(t). 
\]
By Proposition~\ref{pro:sum_piecewise}, 
$F^{e_{i-1}}_{\rm L}(t) + F^{e_i}_{\rm R}(t) - {\rm Opt}(t)$
is a piecewise polynomial function of degree at most two 
with at most $O(N_{\rm Opt})$ pieces. 
Let $N_{v_i}$ be the number of pieces of 
$F^{e_{i-1}}_{\rm L}(t) + F^{e_i}_{\rm R}(t) - {\rm Opt}(t)$ and 
$T^{v_i}_j$ be the interval of the $j$-th piece (from the left) of 
$F^{e_{i-1}}_{\rm L}(t) + F^{e_i}_{\rm R}(t) - {\rm Opt}(t)$. 
Thus, $R(v_i,t)$ is represented as a polynomial function of degree at most two 
on $t \in T^{v_i}_j$ for each $T^{v_i}_j$. 
For each integer $j$ with  $1 \leq j \leq N_{v_i}$, 
by elementary calculation, we obtain the maximum value $G^{v_i}_j$ of $R(v_i, t)$ over $t \in T_j$
in $O(1)$ time. 
By choosing the maximum values among $G^{v_i}_{1}, \ldots, G^{v_i}_{N_{v_i}}$, 
we can obtain $MR(v_i)$ in $O(N_{v_i}) = O(N_{\rm Opt})$ time.
\end{proof}
%}

\subsection{Algorithms and Time Analyses}\label{subsec:algorithm}

Let us give an algorithm that finds a sink location that 
minimizes the maximal regret and 
the analysis of the running time of each step. 

%First, we obtain $F^e(t)$ and $F^v(t)$ for all $e \in E$ and $v \in V$, respectively, 
First, we obtain $F^e_{\rm L}(t)$ and $F^e_{\rm R}(t)$ for all $e \in E$ 
and obtain function ${\rm Opt}(t)$ as a preprocess. 
Applying Lemmas~\ref{lem:func_Phi_merged} and \ref{lem:opt}, 
we take $O(n^2 N_{F} \log n)$ time for these operations. 
Next, we compute $x^{*,e} = \argmin\{ MR(x) \mid x \in e\}$ for all $e \in E$
in $O(n N_{\rm Opt} \alpha(n) \log n)$ time 
by applying Lemma~\ref{lem:MR_e_function}. 
% since $F^{e}(t)$ for all $e \in E$ and ${\rm Opt}(t)$ are available. 
Then, we also compute $MR(v)$ for all $v \in V$ in $O(n N_{\rm Opt})$ time
by applying Lemma~\ref{lem:MR_v_function}. 
Finally, we find an optimal sink location $x^{*}$ in $O(n)$ time
by evaluating the values $MR(x)$ for $x \in \{ x^{*, e} \} \cup V$. 

Since we have $N_{Opt} = O(n N_{F} 2^{\alpha(n)})$, 
the bottleneck of our algorithm is to compute $x^{*,e}$ for all $e \in E$.
Thus, we see that the algorithm runs in 
$O(n^2 N_{F}2^{\alpha(n)} \alpha(n) \log n)$
time, which completes the proof of our main theorem 
because $N_{F} = O(n^2)$, and for the case with uniform edge capacity, $N_{F} = O(n)$.

%%
%% Bibliography
%%

%%%%%%%%%%%%%%%%%%%%%%%%%%%%%%%%%%%%%%%%
% References
%%%%%%%%%%%%%%%%%%%%%%%%%%%%%%%%%%%%%%%%

\bibliographystyle{splncs04}
\bibliography{MMR_AT_Ksink_arXiv}

\begin{thebibliography}{10}
\providecommand{\url}[1]{\texttt{#1}}
\providecommand{\urlprefix}{URL }
\providecommand{\doi}[1]{https://doi.org/#1}

\bibitem{agarwal1989sharp}
Agarwal, P.K., Sharir, M., Shor, P.: Sharp upper and lower bounds on the length
  of general davenport-schinzel sequences. J. Comb. Theory, Ser. {A}
  \textbf{52}(2),  228--274 (1989)

\bibitem{Alstrup2002}
Alstrup, S., Gavoille, C., Kaplan, H., Rauhe, T.: Nearest common ancestors: a
  survey and a new distributed algorithm. In: Proc. of the 14th annual ACM
  symposium on Parallel algorithms and architectures (SPAA 2002). pp. 258--264
  (2002)

\bibitem{arumugam2019minmax}
Arumugam, G.P., Augustine, J., Golin, M.J., Srikanthan, P.: Minmax regret
  k-sink location on a dynamic path network with uniform capacities.
  Algorithmica  \textbf{81}(9),  3534--3585 (2019)

\bibitem{belmonte2015polynomial}
Belmonte, R., Higashikawa, Y., Katoh, N., Okamoto, Y.: Polynomial-time
  approximability of the k-sink location problem. CoRR  \textbf{abs/1503.02835}
  (2015)

\bibitem{Bender2000}
Bender, M.A., Farach-Colton, M.: The {LCA} problem revisited. In: Proc. of the
  4th Latin American Symposium on Theoretical Informatics (LATIN 2000). pp.
  88--94 (2000)

\bibitem{BenkocziBHKK18}
Benkoczi, R., Bhattacharya, B., Higashikawa, Y., Kameda, T., Katoh, N.: Minsum
  $k$-sink problem on dynamic flow path networks. In: Proc. of the 29th
  International Workshop on Combinatorial Algorithms (IWOCA 2018). pp. 78--89
  (2018)

\bibitem{BenkocziBHKK19}
Benkoczi, R., Bhattacharya, B., Higashikawa, Y., Kameda, T., Katoh, N.: Minsum
  $k$-sink problem on path networks. Theor. Comput. Sci.  \textbf{806},
  388--401 (2020)

\bibitem{BhattacharyaGHK17}
Bhattacharya, B., Golin, M.J., Higashikawa, Y., Kameda, T., Katoh, N.: Improved
  algorithms for computing k-sink on dynamic flow path networks. In: Proc. of
  the 15th Workshop on Algorithms and Data Structures (WADS 2017). pp. 133--144
  (2017)

\bibitem{bhattacharya2018n}
Bhattacharya, B., Higashikawa, Y., Kameda, T., Katoh, N.: An ${O}(n^2 \log^2
  n)$ time algorithm for minmax regret minsum sink on path networks. In: Proc.
  of the 29th International Symposium on Algorithms and Computation (ISAAC
  2018) (2018)

\bibitem{bhattacharya2015improved}
Bhattacharya, B., Kameda, T.: Improved algorithms for computing minmax regret
  sinks on dynamic path and tree networks. Theor. Comput. Sci.  \textbf{607},
  411--425 (2015)

\bibitem{chen2016sink}
Chen, D., Golin, M.J.: Sink evacuation on trees with dynamic confluent flows.
  In: 27th International Symposium on Algorithms and Computation (ISAAC 2016)
  (2016)

\bibitem{chen2018minmax}
Chen, D., Golin, M.J.: Minmax centered k-partitioning of trees and applications
  to sink evacuation with dynamic confluent flows. CoRR
  \textbf{abs/1803.09289} (2018)

\bibitem{ford1958}
Ford, L.R., Fulkerson, D.R.: Constructing maximal dynamic flows from static
  flows. Operations research  \textbf{6}(3),  419--433 (1958)

\bibitem{golin2018minmax}
Golin, M.J., Sandeep, S.: Minmax-regret $k$-sink location on a dynamic tree
  network with uniform capacities. CoRR  \textbf{abs/1806.03814} (2018)

\bibitem{hart1986nonlinearity}
Hart, S., Sharir, M.: Nonlinearity of {Davenport--Schinzel} sequences and of
  generalized path compression schemes. Combinatorica  \textbf{6}(2),  151--177
  (1986)

\bibitem{Hershberger89}
Hershberger, J.: Finding the upper envelope of $n$ line segments in ${O}(n \log
  n)$ time. Information Processing Letters  \textbf{33}(4),  169--174 (1989)

\bibitem{Higashikawa14}
Higashikawa, Y.: Studies on the space exploration and the sink location under
  incomplete information towards applications to evacuation planning. PhD
  thesis, Kyoto University, Japan  (2014)

\bibitem{higashikawa2015a}
Higashikawa, Y., Augustine, J., Cheng, S.W., Golin, M.J., Katoh, N., Ni, G.,
  Su, B., Xu, Y.: Minimax regret 1-sink location problem in dynamic path
  networks. Theor. Comput. Sci.  \textbf{588},  24--36 (2015)

\bibitem{higashikawa2018minimax}
Higashikawa, Y., Cheng, S.W., Kameda, T., Katoh, N., Saburi, S.: Minimax regret
  1-median problem in dynamic path networks. Theory Comput. Syst.
  \textbf{62}(6),  1392--1408 (2018)

\bibitem{higashikawa2014f}
Higashikawa, Y., Golin, M.J., Katoh, N.: Minimax regret sink location problem
  in dynamic tree networks with uniform capacity. J. Graph Algorithms Appl.
  \textbf{18}(4),  539--555 (2014)

\bibitem{HigashikawaGK15}
Higashikawa, Y., Golin, M.J., Katoh, N.: Multiple sink location problems in
  dynamic path networks. Theor. Comput. Sci.  \textbf{607},  2--15 (2015)

\bibitem{hoppe2000}
Hoppe, B., Tardos, E.: The quickest transshipment problem. Mathematics of
  Operations Research  \textbf{25}(1),  36--62 (2000)

\bibitem{kouvelis1997}
Kouvelis, P., Yu, G.: Robust Discrete Optimization and its Applications. Kluwer
  Academic Publishers, London (1997)

\bibitem{li2016a}
Li, H., Xu, Y.: Minimax regret 1-sink location problem with accessibility in
  dynamic general networks. Eur. J. Oper. Res.  \textbf{250}(2),  360--366
  (2016)

\bibitem{li2016b}
Li, H., Xu, Y., Ni, G.: Minimax regret vertex 2-sink location problem in
  dynamic path networks. Journal of Combinatorial Optimization  \textbf{31}(1),
   79--94 (2016)

\bibitem{mamada2006}
Mamada, S., Uno, T., Makino, K., Fujishige, S.: An ${O}(n \log^2 n)$ algorithm
  for a sink location problem in dynamic tree networks. Discret. Appl. Math.
  \textbf{154},  2387--2401 (2006)

\bibitem{skutella2009}
Skutella, M.: An introduction to network flows over time. In: Research Trends
  in Combinatorial Optimization, pp. 451--482. Springer (2009)

\bibitem{vairaktarakis1999incorporation}
Vairaktarakis, G.L., Kouvelis, P.: Incorporation dynamic aspects and
  uncertainty in 1-median location problems. Naval Research Logistics (NRL)
  \textbf{46}(2),  147--168 (1999)

\end{thebibliography}

\end{document}